\documentclass[11pt]{article}
\usepackage[T1]{fontenc}
\usepackage[utf8]{inputenc}
\usepackage{amsmath,amssymb,amsthm}
\usepackage{aliascnt}
\usepackage[margin=1in]{geometry}
\usepackage{microtype}
\usepackage{xcolor}
\usepackage{enumitem}
\usepackage{url}
\usepackage[colorlinks=true,linkcolor=blue!55!black,citecolor=blue!55!black,urlcolor=blue!55!black]{hyperref}
\usepackage[nameinlink,noabbrev,capitalise]{cleveref}
\newtheorem{theorem}{Theorem}[section]
\newaliascnt{lemma}{theorem}
\newtheorem{lemma}[lemma]{Lemma}
\aliascntresetthe{lemma}
\newaliascnt{proposition}{theorem}
\newtheorem{proposition}[proposition]{Proposition}
\aliascntresetthe{proposition}
\newaliascnt{corollary}{theorem}
\newtheorem{corollary}[corollary]{Corollary}
\aliascntresetthe{corollary}
\theoremstyle{definition}
\newaliascnt{definition}{theorem}

\aliascntresetthe{definition}
\theoremstyle{remark}
\newaliascnt{remark}{theorem}

\aliascntresetthe{remark}
\numberwithin{equation}{section}
\DeclareMathOperator{\Adv}{Adv}
\DeclareMathOperator{\Ent}{Ent}
\DeclareMathOperator{\Sym}{Sym}
\DeclareMathOperator{\He}{He}
\DeclareMathOperator{\GOE}{GOE}
\newcommand{\VSP}{\mathrm{VSP}}
\newcommand{\E}{\mathbb E}

\newcommand{\HS}{\mathrm{HS}}

\setlist[itemize]{leftmargin=1.5em,itemsep=2pt,topsep=4pt}
\setlist[enumerate]{leftmargin=1.8em,itemsep=2pt,topsep=4pt}
\allowdisplaybreaks[1]
\hypersetup{pdftitle={Tight Bounds on the Randomized Communication Complexity of Vector-in-Subspace},pdfauthor={Chengu Wang}}
\title{Tight Bounds on the Randomized\\
Communication Complexity of Vector-in-Subspace}
\author{Chengu Wang\\\small\href{mailto:wangchengu@gmail.com}{\texttt{wangchengu@gmail.com}}}
\date{}
\begin{document}
\maketitle
\begin{abstract}
  We prove that the randomized communication complexity of Vector-in-Subspace
  in $\mathbb R^n$ with error at most $1/3$ is $\Theta(\sqrt n)$.
  Alice holds a unit vector $u$ and Bob
  holds an $n/2$-dimensional subspace $H$; they must distinguish $u\in H$ from
  $u\in H^\perp$. An exact one-way quantum protocol uses $O(\log n)$ qubits,
  making the problem a geometric example of an exponential quantum advantage
  over interactive classical communication.
  Our lower bound matches Raz's $O(\sqrt n)$ upper bound, improves the
  $\Omega(n^{1/3})$ lower bound of Klartag and Regev, and removes the total-rank
  restriction in Grupel's $\Omega(\sqrt n)$ result.

  The proof establishes a Gaussian moment bound over complete rectangle
  partitions. A complex Gaussian representation yields a vanishing even
  moment whose cancellation is preserved when averaging over the partition.
  We control the signed conditional moments around a real center by
  decomposing conditional Gaussian measures into fixed coordinates and
  sub-Gaussian components. Combining this estimate with spherical harmonic
  analysis gives the communication lower bound.
\end{abstract}

\section{Introduction}\label{sec-intro}

Communication complexity, introduced by Yao~\cite{Yao79},
measures the communication needed to compute a
function whose input is split between two parties~\cite{KN97}. Alice
receives $x$ and Bob receives $y$, and they exchange messages so that one
of them can output $f(x,y)$. Each message may depend on the sender's input
and the preceding messages. Local computation is unrestricted; the cost
is the total number of bits exchanged, maximized over inputs and protocol
executions. A deterministic protocol must always give the correct answer.
For a promise problem, correctness is required only on inputs satisfying
the promise.

In \emph{randomized communication}, the parties may also use
random bits, independent of their inputs.
A bounded-error protocol must output the correct answer with probability
at least $2/3$ on every fixed promised input, where the probability is over
the protocol's randomness. We write $R_{1/3}(f)$ for the minimum worst-case
communication cost of such a protocol with public randomness.
The messages may depend on the public random bits as well as
the inputs and earlier messages.

In \emph{quantum communication}, introduced by Yao~\cite{Yao93}, the
parties exchange quantum states and perform local quantum operations.
The communication cost counts the qubits transmitted. Bounded-error
quantum protocols satisfy the same correctness requirement on every
input, with probability also taken over measurement outcomes. In either
model, a \emph{one-way} protocol consists of a single message from Alice
to Bob, who then outputs the answer; an \emph{interactive} protocol may
exchange messages in both directions.

How much classical communication is needed to replace a single quantum
message?
The Vector-in-Subspace problem \cite{kremer,raz} shows an exponential separation
between classical and quantum communication complexity \cite{klartag-regev}.
Alice is given a unit vector $x\in S^{n-1}$, and Bob is given a
subspace $H\subseteq\mathbb R^n$ of dimension $n/2$. They are promised that
either $x\in H$ or $x\in H^\perp$, and must determine which alternative holds.
We denote this promise problem by $\VSP_n$ and take $n$ to be even.

\subsection{Previous results}

A quantum protocol for $\VSP_n$ is immediate: Alice sends the state with amplitude
vector $x$, using $\lceil\log_2 n\rceil$ qubits, and Bob measures the two
orthogonal subspaces. The answer is correct with probability one.

Raz~\cite[Theorem~3.8]{raz} gave an $O(\sqrt n)$-bit classical protocol
for a constant-gap version of Vector-in-Subspace. The exact promise is a
restriction of that problem, so the same upper bound applies to $\VSP_n$.
At a high level, public randomness supplies an exponentially large
collection of Gaussian vectors. Alice sends the index of the vector with
the largest inner product with her input, and Bob compares its distances
to his subspace and its orthogonal complement. The index has
$O(\sqrt n)$ bits, although
finding it may require substantial local computation.
Klartag and Regev \cite[Section~4]{klartag-regev} describe a spherical
version of this protocol;
Montanaro \cite[Appendix~A]{montanaro} gives a detailed analysis of the
closely related projector-estimation formulation. Thus the square-root
scale was already available as an upper bound, and the question was whether
interaction could substantially reduce it.

The distinction between one-way and interactive classical communication is
essential. Gavinsky, Kempe, Kerenidis, Raz, and de Wolf \cite{gkkrw}
proved an exponential separation between one-way quantum and one-way
classical communication for a partial Boolean function based on partial
matching. Their problem has a short interactive classical protocol, so
that separation alone does not establish hardness against interaction.
As explained explicitly by Gosset and Smolin \cite{gosset-smolin}, their
partial-matching bound also gives an $\Omega(\sqrt n)$ one-way lower bound
for a constant-gap projector problem: distinguish
$\langle x,Px\rangle\le 1/3$ from $\langle x,Px\rangle\ge 2/3$.
This is a broader promise than membership in $H$ or $H^\perp$.
Gosset and Smolin also construct classical sketches of quantum states that
give a computationally efficient one-way protocol with
$\widetilde O(\sqrt n)$ communication. Their result improves the local
computational cost of the classical simulation while matching its
communication scale up to logarithmic factors.

On the lower-bound side, Klartag and Regev \cite{klartag-regev} proved the first polynomial lower
bound for classical protocols for $\VSP_n$, namely
$\Omega(n^{1/3})$. Combined with the quantum protocol above, their result
established an exponential separation between one-way quantum and
interactive classical communication for a partial function. Their argument
studies how a random half-dimensional subspace samples a measurable set on
the sphere. It combines spherical hypercontractivity, the Radon transform,
and a martingale argument, and then applies a rectangle bound to a
communication protocol.
Kerenidis, Laplante, Lerays, Roland, and Xiao \cite{kllrx} subsequently
showed that the lower-bound method also yields an
$\Omega(n^{1/3})$ bound on classical information complexity, using their
relaxed partition bound and zero-communication compression framework.

There is a concrete obstruction to strengthening the particular sampling
estimate in the Klartag--Regev proof. A spherical cap of measure
$\exp(-\Theta(n^{1/3}))$ can have a constant relative sampling error with
probability $\exp(-\Theta(n^{1/3}))$ \cite[Section~2]{klartag-regev}.
This makes that one-subspace concentration statement sharp at the
cube-root scale. Klartag and Regev therefore suggested studying $H$ and
$H^\perp$ together, with a lower bound on the arithmetic mean of their
intersection measures. The cap example concerns this specific sampling
statement; it does not preclude stronger arguments that use the complete
partition produced by a protocol.

Grupel \cite{grupel} developed the paired-subspace approach and obtained
an $\Omega(\sqrt n)$ lower bound for protocols of total rank at most
$\alpha\sqrt n$, for an absolute constant $\alpha>0$.
Here the rank of an Alice decision is the number of linear functionals of
her vector on which the decision depends, and total rank sums these
quantities along a protocol execution. This is a restriction on how Alice
uses her input, rather than the rank of a communication matrix.
Grupel proved a paired concentration theorem for functions depending on
$O(\sqrt n)$ directions and conjectured a corresponding theorem for
arbitrary measurable sets, using the geometric mean of the two
intersection measures. That conjecture would imply the unrestricted
square-root communication bound.

\subsection{Our result}

We prove a square-root lower bound on the randomized communication complexity
of Vector-in-Subspace.

\begin{theorem}\label{thm-main}
  $ R_{1/3}(\VSP_n)=\Theta(\sqrt n).$
\end{theorem}

We improve the $\Omega(n^{1/3})$ lower bound of Klartag and
Regev~\cite{klartag-regev} to $\Omega(\sqrt n)$, matching Raz's
$O(\sqrt n)$ upper bound~\cite{raz}. Our lower bound also removes the
total-rank restriction in Grupel's square-root lower bound~\cite{grupel}.

\subsection{Proof overview}

The proof has one central estimate, concerning a complete partition of a
product Gaussian space into rectangles. A deterministic protocol using
$d$ bits, with its output included in the transcript, gives at most $2^d$
such rectangles. Replace Alice's spherical input by the direction of a
standard Gaussian vector $Y$, and Bob's subspace by the top-half
eigenspace of an independent Gaussian orthogonal ensemble matrix $G$.
These maps preserve the product reference distribution and the number
of rectangles.

The spherical harmonic reduction associates to each level $k$ a Gaussian
polynomial $F_k(Y,G)$, obtained by replacing the vector and matrix powers
in $(Y^TGY)^k$ by their Hermite, or Wick, versions. Our central bound is
\begin{equation}\label{eq:cp-intro}
  \mathbb E\bigl|\mathbb E[F_k\mid\mathsf T]\bigr|
  \le C^k k!\,(d+k^2)^k
\end{equation}
for every complete rectangle partition $\mathsf T$ with at most $2^d$
cells. It holds for complete partitions whether or not they arise from
a protocol tree. The expectation weights each cell by its probability
under the product Gaussian law.

\paragraph{Cancellation before absolute values.}
Let $Z$ and $G'$ be independent copies of $Y$ and $G$, independent also
of the inputs, and put
\[
  Q=\frac12(Y+iZ)^T(G+iG')(Y+iZ).
\]
The imaginary Gaussians perform the Wick corrections, giving
$F_k=2^k\mathbb E[Q^k \mid Y,G]$.
If we first choose a rectangle according to its probability and then
sample the inputs conditioned on that rectangle, completeness restores
the original Gaussian law. In particular,
\[
  \mathbb E Q^{2k}=0.
\]
This identity is global: it need not hold in any one rectangle.
We construct a random real center $\alpha$ and control the signed
conditional moments of $Q-\alpha$. Expanding the vanishing even moment
then bounds $\mathbb E\alpha^{2k}$, since this leading term is
nonnegative. A second expansion, now of $Q^k$, proves
\eqref{eq:cp-intro}. This is the cancellation used throughout the proof.

\paragraph{A real center and a controlled remainder.}
On a fixed rectangle, Alice's conditional law has density bounded by
$e^a$ relative to Gaussian measure, where $a$ is the logarithm of the
reciprocal of her cell's probability. We express this law as a mixture
whose components fix some coordinates at a vector $y$ and leave a
sub-Gaussian law on the orthogonal complement. The sum of the residual
log-density bound, $|y|^2$, and the number of fixed coordinates has an
exponential tail controlled by $a$.
Let $E$ be the subspace of fixed coordinates and $G_{EE}$ the corresponding
block of $G$. These coordinates provide the real center
$\alpha=(y^T\,G_{EE}\,y-\operatorname{tr}G_{EE})/2$.
Their fluctuations satisfy elementary Gaussian bounds. The remaining
terms are controlled through signed Wick moments of the sub-Gaussian
component. Taking absolute moments of the complex variable too early
would introduce an unwanted dependence on the ambient dimension.

\paragraph{The moment estimate needed for the remainder.}
The bulk estimate requires both partition norms and a stronger Frobenius
bound for its Wick tensors. One polynomial-moment induction gives both.
Its Gaussian input comes from a derivative expansion stopped after
roughly half the tensor slots. The terms evaluated in $L^2$ retain a
Hilbert--Schmidt norm and cost only $q^{r/4}$ at Gaussian degree $r$;
the remaining terms are bounded by completing the unevaluated Gaussian
slots using partition norms. Keeping these two contributions separate
allows the induction to close with a quadratic inflation in the degree,
which is the $k^2$ in \eqref{eq:cp-intro}.

\paragraph{Returning to the sphere.}
At harmonic level $k$, the transfer loses at most
$C^k/(k!n^{k/2})$. For $k\le\sqrt d$, \eqref{eq:cp-intro} therefore
bounds the level's contribution by $(Cd/\sqrt n)^k$.
For $\sqrt d<k\le d$, a simpler hypercontractive bound gives the same
estimate. The exact subspace-transform spectrum controls levels above
$d$. Summing these bounds makes the protocol's advantage small when
$d$ is a sufficiently small multiple of $\sqrt n$.

\subsection{Connections to earlier analytic methods}

Our spherical reduction uses the harmonic-analysis framework underlying
Klartag--Regev \cite{klartag-regev}. We give the spectrum, the Gaussian
lift, its leading coefficient, and the smoothing argument explicitly,
including the constants needed uniformly in the harmonic degree.

The decomposition into affine slices is related to the adaptive Gaussian
cleanup of Girish, Sinha, Tal, and Wu \cite{gstw}. They control directional
moments by revealing problematic linear or quadratic forms during a
protocol and use this to bound Fourier growth of XOR-fibers.
Here we decompose a conditional law using failing one-dimensional tail
tests. The conclusion supplies a sub-Gaussian bound in every direction
on each terminal component, together with exponential control of its
combined density, offset, and dimension cost.
Entropy-efficient decompositions with covariance control also appear in
Eldan's work \cite{eldan}; those use exponential tilts and a different
entropy--covariance tradeoff. These results motivate a useful distinction:
covariance control alone does not supply the all-moment estimate used
for our bulk components.

The tensor norms in the Gaussian estimates are the partition norms of
Lata{\l}a's theory of Gaussian chaoses \cite{latala}; see also Lehec's
simplified proof \cite{lehec-chaos}. Our one-Gaussian-slot inequality
adapts the Gaussian-shift, PAC-Bayesian proof of Aden-Ali
\cite[Theorem~1.2]{aden-ali} directly to $L^p$ moments.
The derivative step is the classical log-Sobolev moment method of
Aida--Stroock \cite{aida-stroock}, related to the entropy method of
Boucheron, Bousquet, Lugosi, and Massart \cite{bblm} and to the iterated
derivative bounds of Adamczak--Wolff \cite{adamczak-wolff}.
We retain proofs of the forms we use to track their dependence on the
degree. In particular, the split estimate is tailored to the two tensor
bounds in the bulk induction; we do not require an improvement of the
general Lata{\l}a estimate at every moment.

\paragraph{Organization.}
\Cref{sec-prelim} fixes the model and Gaussian notation.
\Cref{sec-transfer} proves the spherical reduction and derives
\cref{thm-main} from the complete-partition estimate.
\Cref{sec-cancel} proves that estimate from the decomposition and bulk
lemmas. Their proofs follow in \cref{sec-decomp,sec-chaos,sec-bulk}.

\section{Preliminaries}\label{sec-prelim}

\subsection{Protocols, measures, and constants}
Let $\sigma_{\rm u}$ denote uniform probability measure on $S^{n-1}$ and
$\nu$ Haar probability measure on the Grassmannian of $n/2$-dimensional
subspaces. Write $\mu_0=\sigma_{\rm u}\otimes\nu$ and $\E_0$ for its
expectation. Let $\mu_+$ and $\mu_-$ be the promise laws obtained by
first sampling $H\sim\nu$ and then sampling $x$ uniformly on
$S^{n-1}\cap H$ or $S^{n-1}\cap H^\perp$, respectively. The proof uses
the fair promise law $(\mu_++\mu_-)/2$.
For an output $\psi\in\{-1,1\}$, with $+1$ meaning $x\in H$,
\[
  \Adv(\psi)=\frac{\E_{\mu_+}\psi-\E_{\mu_-}\psi}{4}.
\]
Its success probability under this law is $1/2+\Adv(\psi)$.

A deterministic cost-$d$ protocol induces a complete partition of the
product input space into at most $2^d$ measurable rectangles. Whenever
we use its output as a function of the rectangle, we include the output
bit in the transcript. The extra bit will be accounted for at the end.
A protocol specified on promised inputs can be truncated at its
worst-case cost and completed on other inputs. Fixing all public and
private random coins gives a deterministic protocol. The joint Borel
assumption ensures that these operations and the required expectations
are measurable. These are the usual rectangle and distributional
reductions; see, for example, \cite[Section~4]{klartag-regev}.

The letters $C,c$ denote positive absolute constants whose values may
change between occurrences. A constant in $C^k$ is independent of $k$
and the dimension. Logarithms are natural, except for the number of
communication bits. We use $B_1=5$ for the sub-Gaussian constant.

\subsection{Gaussian and tensor conventions}
Write $\gamma_W$ for standard Gaussian measure on a finite-dimensional
Euclidean space $W$. The law $\GOE(n)$ is standard Gaussian measure on
$\Sym(\mathbb R^n)$ with inner product $\operatorname{tr}(AB)/2$.
Thus its independent off-diagonal entries have variance one and its
diagonal entries have variance two. A matrix's usual Frobenius norm is
$\sqrt2$ times its norm in this GOE Euclidean space.

For a standard Gaussian $Z$ on $W$, define the Hermite tensor
\[
  \He_r(y)=\E_Z(y+iZ)^{\otimes r}.
\]
This classical representation follows by expanding the generating
function $\E_Z e^{\langle t,y+iZ\rangle}
=e^{\langle t,y\rangle-|t|^2/2}$.
For a standard Gaussian on any Euclidean space we use its own inner
product to define its Hermite tensors. In particular,
$\He_r^{\GOE}(G)=\E(G+iG')^{\otimes r}$.
Pairings of complex tensors are bilinear; their norms are Hermitian.
An unqualified tensor norm is the Hilbert--Schmidt norm.
The map $\mathcal J_m$ replaces each matrix slot by its two coordinate
indices and symmetrizes those vector indices. Hence
\[
  \langle w^{\otimes 2m},\mathcal J_m S^{\otimes m}\rangle=(w^TSw)^m,
  \qquad \|\mathcal J_m T\|_{\HS}\le 2^{m/2}\|T\|_{\GOE^{\otimes m}}.
\]

We use the classical Wick isometry and Gaussian hypercontractivity
\cite{isserlis,nelson,janson}:
\begin{equation}\label{eq:classical-chaos}
  \E\langle A,\He_r(Z)\rangle^2=r!\|\Sym A\|_{\HS}^2,
  \qquad \|P\|_q\le(q-1)^{r/2}\|P\|_2\quad(q\ge2),
\end{equation}
where $\Sym A$ denotes symmetrization, $A$ is real, and $P$ belongs
to the $r$th Gaussian chaos.
Complex versions follow by treating real and imaginary parts separately.
The Gaussian log-Sobolev inequality \cite{gross} is
$\Ent(f^2)\le2\E|\nabla f|^2$.

For $A\in H_1\otimes\cdots\otimes H_r$ and a partition $J$ of $[r]$,
the partition norm is
\begin{equation}\label{eq:partition-norm}
  \|A\|_J=\sup\left\{\left|\left\langle A,
    \bigotimes_{g\in J}y^g\right\rangle\right|:
  y^g\in\bigotimes_{i\in g}H_i,\quad\|y^g\|_{\HS}\le1\right\}.
\end{equation}
Write $\mathcal P(r)$ for all partitions of $[r]$.
These are the norms used by Lata{\l}a \cite{latala} and Lehec
\cite{lehec-chaos}. The one-block norm is the Hilbert--Schmidt norm,
also called the Frobenius norm; the all-singletons norm is the injective
norm. For matrices, these reduce to the Frobenius and operator norms.
For complex tensors, the test tensors in \eqref{eq:partition-norm} are
complex and have Hermitian norm at most one.

A probability law on $W$ is $\mathrm{SG}(B)$ if
$\|\langle v,Y\rangle\|_q\le B\sqrt q$ for every unit $v\in W$
and every real $q\ge1$. This is an uncentered condition. When the law
has density $\varphi$ relative to $\gamma_W$, we also call $\varphi$
$\mathrm{SG}(B)$ and write $\sigma=\log\operatorname{ess\,sup}\varphi$.
All suprema of densities below mean essential suprema. On a
zero-dimensional space, the Gaussian law is a point mass and this
condition is vacuous.

\subsection{Change of density and the probability of a cell}
If $\rho\le e^a$ is a probability density relative to a probability
measure $\eta$, then, for $1\le q\le s<\infty$,
\begin{equation}\label{eq:density}
  \|h\|_{L^q(\rho\eta)}\le e^{a/s}\|h\|_{L^s(\eta)}.
\end{equation}
For $q<s$, apply H\"older with exponents $s/q$ and $s/(s-q)$, and use
$\int\rho^{s/(s-q)}d\eta\le e^{aq/(s-q)}$.
For $q=s$, use $\rho\le e^a$ directly. We also use
\begin{equation}\label{eq:quadratic}
  t^2\le ut+v\quad\Longrightarrow\quad t\le u+\sqrt v
  \qquad(t,u,v\ge0).
\end{equation}

For a rectangle $R=A\times\mathcal C$ under
$\gamma_n\otimes\GOE(n)$, put
$a=\log(1/\gamma_n(A))$ and $b=\log(1/\Pr(G\in\mathcal C))$.
We call these quantities the surprises of its two parts.
The rectangle's surprise is $a+b$.
\begin{lemma}[Surprise of a complete partition]\label{lem-surprise}
  Let $\mathsf T$ be a complete partition into at most $2^d$ rectangles, and let $a_{\mathsf T},b_{\mathsf T}$ be the surprises of the random rectangle containing $(Y,G)$. Then $\Pr(a_{\mathsf T}+b_{\mathsf T}>d\log2+v)\le e^{-v}$ for $v\ge0$. Hence $\|a_{\mathsf T}+b_{\mathsf T}+u\|_q\le d\log2+u+q$ for all $u\ge0$ and $q\ge1$.

\end{lemma}
\begin{proof}
  $a_{\mathsf T}+b_{\mathsf T}>d\log2+v$ means the rectangle has probability below $2^{-d}e^{-v}$. At most $2^d$ rectangles qualify, with total mass below $e^{-v}$. The excess over $d\log2$ is therefore dominated by an exponential variable, whose $L^q$ norm is $\Gamma(q+1)^{1/q}\le q$.
\end{proof}

\subsection{The Gaussian kernels}
Let $Y\sim\gamma_n$ and $G\sim\GOE(n)$ be independent. Define
\begin{equation}\label{eq:kernel}
  F_k(Y,G)=\langle\He_{2k}(Y),\mathcal J_k\He_k^{\GOE}(G)\rangle,
  \qquad
  \mathcal B_k(d)=\sup_{\mathsf T}\E\bigl|\E[F_k\mid\mathsf T]\bigr|,
\end{equation}
where $\mathsf T$ ranges over complete measurable rectangle partitions
with at most $2^d$ cells. Null cells can always be omitted.
\begin{lemma}[Lift identity]\label{lem-lift}

  $F_k(y,G)=\mathbb E_{Z,G'}\big[(w^T(G+iG')w)^k\big]$ with $w=y+iZ$, where $Z\sim\gamma_n$ and $G'\sim\mathrm{GOE}(n)$ are independent. In particular, $F_k$ is real. Under independent $Y\sim\gamma_n$ and $G\sim\mathrm{GOE}(n)$, it is a pure chaos of degrees $(2k,k)$ with mean zero.

\end{lemma}
\begin{proof}
  The Hermite representations give
  \[
    \mathrm{He}_{2k}(y)=\mathbb E_Z w^{\otimes2k},\qquad
    \mathrm{He}^{\rm GOE}_k(G)=\mathbb E_{G'}(G+iG')^{\otimes k}.
  \]
  For a symmetric matrix $S$, the definition of $\mathcal J_k$ gives
  $\langle w^{\otimes2k},\mathcal J_kS^{\otimes k}\rangle=(w^TSw)^k$.
  Apply Fubini.
\end{proof}

For example, $F_1(y,G)=y^TGy-\operatorname{tr}G$: the correction subtracts the Gaussian mean. In higher degrees, averaging over $(Z,G')$ performs all the corresponding Wick corrections. This is the reason for introducing the imaginary Gaussians.

\section{Spherical reduction and the main theorem}\label{sec-transfer}

We reduce the communication lower bound to the following estimate.
Its proof occupies \cref{sec-cancel} and uses the analytic inputs proved
in \cref{sec-decomp,sec-chaos,sec-bulk}.

\begin{theorem}[Complete-partition estimate]\label{thm-cp}
  For every complete partition $\mathsf T$ of
  $\mathbb R^n\times\Sym(\mathbb R^n)$ into at most $2^d$ measurable
  rectangles and every integer $k\ge1$,
  \begin{equation}\label{eq:cp}
    \E\bigl|\E[F_k(Y,G)\mid\mathsf T]\bigr|
    \le C^k k!(d+k^2)^k.
  \end{equation}
  The expectation is under the independent laws $\gamma_n$ and $\GOE(n)$.
\end{theorem}

\subsection{The subspace-transform spectrum}\label{sec-spectrum}

The use of the spherical Radon transform in communication lower bounds
follows Klartag--Regev~\cite{klartag-regev}. We give the spectrum and
normalization for half-dimensional subspaces explicitly. Standard facts
about spherical harmonics and Gegenbauer polynomials can be found in
\cite{groemer,szego}.

For $n\ge4$, put $\lambda=(n-2)/4$, $S_H=2P_H-I$, and define
\begin{equation}\label{eq:spherical-kernel}
  K_0=1,\qquad
  K_k(x,H)=\frac{\lambda+k}{\lambda}C_k^\lambda(x^TS_Hx),
  \qquad
  \mathfrak a_k=\frac{(1/2)_k}{((n-1)/2)_k}\quad(k\ge1).
\end{equation}
Here $(z)_k=z(z+1)\cdots(z+k-1)$, and $C_k^\lambda$ is the
Gegenbauer polynomial. Let
\[
  \mathcal Rf(H)=\int_{S^{n-1}\cap H}f\,d\sigma_H,
  \qquad \mathcal R^-f(H)=\mathcal Rf(H^\perp).
\]
For an $L^2(\sigma_{\rm u})$ function $f$, denote its component in
degree $j$ spherical harmonics by $f_j$.

\begin{lemma}[Spectrum and kernel normalization]\label{lem-spectrum}
  For every fixed $H$, $K_k(\cdot,H)$ is a spherical harmonic of degree $2k$.
  The transforms annihilate odd degrees. For degree $2k$,
  \begin{align}
    \|\mathcal Rf_{2k}\|_2^2
    &=\mathfrak a_k\|f_{2k}\|_2^2,
    &\mathcal R^-f_{2k}&=(-1)^k\mathcal Rf_{2k},\label{eq:spectrum}\\
    \mathcal Rf_{2k}(H)
    &=\int f_{2k}(x)K_k(x,H)\,d\sigma_{\rm u}(x).
    \label{eq:kernel-normalization}
  \end{align}
  Images of distinct degrees are orthogonal. The sequence $\mathfrak a_k$
  decreases and satisfies
  \begin{equation}\label{eq:spectral-tail}
    \mathfrak a_k\le\left(\frac{2k}{n-1}\right)^k.
  \end{equation}
\end{lemma}

\begin{proof}
  For general $1\le m\le n-1$, let $\mathcal R_m$ average over
  $m$-dimensional subspaces. Jensen and rotational invariance show that
  it is an $L^2$ contraction. The operator $\mathcal R_m^*\mathcal R_m$
  is rotation invariant, hence scalar on each spherical harmonic degree.
  The odd degrees vanish by symmetry.

  For a fixed unit vector $x$, the transition represented by
  $\mathcal R_m^*\mathcal R_m$ chooses a uniform $m$-subspace containing
  $x$, then a uniform unit vector $y$ in that subspace. Therefore,
  with $t=\langle x,y\rangle$,
  \[
    \E t^{2j}=\frac{(1/2)_j}{(m/2)_j}.
  \]
  The normalized zonal harmonic of degree $2k$ has the expansion
  \[
    Z_{2k}(t)=\frac{(-1)^k(1/2)_k}{((n-1)/2)_k}
    \sum_{j=0}^k
    \frac{(-k)_j(k+(n-2)/2)_j}{(1/2)_j\,j!}\,t^{2j},
    \qquad Z_{2k}(1)=1.
  \]
  Its expected value under this transition is the eigenvalue. The
  terminating Vandermonde identity
  \[
    \sum_{j=0}^k\frac{(-k)_j(b)_j}{(c)_j\,j!}
    =\frac{(c-b)_k}{(c)_k}
  \]
  gives
  \begin{equation}\label{eq:general-spectrum}
    \lambda_{2k}(n,m)=
    \frac{(1/2)_k((n-m)/2)_k}{((n-1)/2)_k(m/2)_k},
  \end{equation}
  using $(1-k-(n-m)/2)_k=(-1)^k((n-m)/2)_k$.
  At $m=n/2$ this is $\mathfrak a_k$.
  The product
  $\mathfrak a_k=\prod_{j=0}^{k-1}(2j+1)/(n-1+2j)$
  proves monotonicity and \eqref{eq:spectral-tail}.

  The cross operator $\mathcal R^*\mathcal R^-$ samples a uniform
  orthogonal pair $x,y$. Its eigenvalue on degree $2k$ is
  $Z_{2k}(0)=(-1)^k\mathfrak a_k$. Since both self-operators have
  eigenvalue $\mathfrak a_k$, for a degree-$2k$ harmonic $f$ we have
  \[
    \|\mathcal R^-f-(-1)^k\mathcal Rf\|_2^2
    =2\mathfrak a_k\|f\|_2^2
    -2(-1)^k(-1)^k\mathfrak a_k\|f\|_2^2=0.
  \]
  The diagonalization of $\mathcal R^*\mathcal R$ also proves
  orthogonality of the images of distinct degrees.

  For the kernel normalization, fix $H$ and sample $x$ uniformly on the
  sphere. The variable $v=\|P_Hx\|^2$ has beta parameters $n/4,n/4$,
  so $t=2v-1$ has density proportional to
  $(1-t^2)^{n/4-1}$ on $[-1,1]$. Its orthogonal polynomials are
  $C_k^\lambda$. On the sphere, the $O(H)\times O(H^\perp)$-invariant
  polynomial functions of degree at most $2k$ are precisely the
  polynomials in $t$ of degree at most $k$, a space of dimension $k+1$.
  Harmonic decomposition preserves invariance, and odd degrees contain
  no invariants because $-I$ lies in the group. Hence the invariant
  harmonic space in each even degree is one-dimensional. In degree
  $2k$, its orthogonality to lower-degree polynomials in $t$ identifies
  it with the span of $C_k^\lambda(t)$. Thus $K_k(\cdot,H)$ is a degree-$2k$
  spherical harmonic, and averaging any harmonic of this degree over
  $O(H)\times O(H^\perp)$ gives a multiple of $C_k^\lambda(t)$.
  The classical Gegenbauer norm and endpoint formulas~\cite{szego} give
  \begin{equation}\label{eq:gegenbauer-norm}
    \E[C_k^\lambda(t)^2]
    =\frac{\lambda}{\lambda+k}C_k^\lambda(1),
    \qquad C_k^\lambda(1)=\frac{(2\lambda)_k}{k!}.
  \end{equation}
  The subspace average of the invariant polynomial is its value at
  $t=1$. Taking its coefficient by orthogonal projection and applying
  \eqref{eq:gegenbauer-norm} gives \eqref{eq:kernel-normalization}.
  All identities extend from finite harmonic sums to $L^2$ by contraction.
\end{proof}

\begin{lemma}[Harmonic truncation]\label{lem-trunc}
  If $\psi\in\{-1,1\}$ is constant on a complete partition into at most
  $2^d$ rectangles, then, for every integer $J\ge1$,
  \begin{equation}\label{eq:truncation}
    \Adv(\psi)\le\frac12
    \sum_{\substack{1\le k\le J\\k\text{ odd}}}|\E_0[\psi K_k]|
    +\frac12\sqrt{2^d\mathfrak a_{J+1}}.
  \end{equation}
\end{lemma}

\begin{proof}
  On a rectangle $A\times\mathcal C$, the difference between the two
  promise probabilities is
  $\int_{\mathcal C}(\mathcal R-\mathcal R^-)\mathbf1_A\,d\nu$.
  Since $K_k(\cdot,H)$ is a degree-$2k$ harmonic, orthogonality allows
  us to replace $(\mathbf1_A)_{2k}$ by $\mathbf1_A$ in
  \eqref{eq:kernel-normalization}. Thus, by \cref{lem-spectrum}, the
  contribution through degree $2J$ is
  $2\sum_{k\le J,\,k\text{ odd}}
  \int_{A\times\mathcal C}K_k\,d\mu_0$.
  The omitted operator has norm at most $2\sqrt{\mathfrak a_{J+1}}$.
  Its contribution on this rectangle is therefore at most
  $2\sqrt{\mathfrak a_{J+1}}
  \sqrt{\sigma_{\rm u}(A)\nu(\mathcal C)}$ in absolute value.
  The rectangle masses sum to one. Cauchy--Schwarz over at most $2^d$
  rectangles bounds the error in the signed output difference by
  $2\sqrt{2^d\mathfrak a_{J+1}}$. Dividing by four gives the result.
  These are $L^2$ expansions of indicator functions followed by bounded
  operators and finitely many integrals; no pointwise density expansion
  of the singular promise law is needed.
\end{proof}

\subsection{The Gaussian lift and its leading coefficient}

Let $X=Y/|Y|$, and let $H(G)$ be the span of the eigenvectors belonging
to the largest $n/2$ eigenvalues of $G$. Ties form a null set and can
be resolved by any fixed Borel convention. Orthogonal invariance makes
$H(G)$ Haar, independently of the ordered eigenvalues. Thus
$(X,H(G))$ has law $\mu_0$, and pulling a protocol back under these
two separate maps preserves its cost and rectangle structure.

\begin{lemma}[Triangular expansion and leading coefficient]\label{lem-top}
  There are real coefficients $\alpha_{kj}$ such that
  \begin{equation}\label{eq:triangular}
    \E[F_k(Y,G)\mid X,H]=\sum_{j=0}^k\alpha_{kj}K_j(X,H).
  \end{equation}
  The leading coefficient obeys
  \begin{equation}\label{eq:top-coefficient}
    \alpha_{kk}\ge k!\,n^{k/2}6^{-k/2}.
  \end{equation}
\end{lemma}

\begin{proof}
  Average over the independent Gaussian radius and, conditional on $H$,
  over the remaining matrix data. The result is a stabilizer-invariant
  even polynomial on the sphere of degree at most $2k$. It is therefore
  a polynomial of degree at most $k$ in $t=X^TS_HX$, proving
  \eqref{eq:triangular}.

  The top spherical component equals that of $(Y^TGY)^k$.
  Indeed, the leading term of $\He_{2k}(Y)$ is $Y^{\otimes2k}$,
  and its lower terms have degree less than $2k$.
  Every matrix Hermite correction contracts a pair of matrix slots
  against the GOE covariance. Contracting this covariance against four
  copies of $Y$ gives $2|Y|^4$, so it introduces vector traces and has
  no degree-$2k$ spherical component.

  It remains to lower-bound the leading coefficient of the conditional
  polynomial for $(X^TGX)^k$. Write $n=2m$, $\Gamma=G/\sqrt n$,
  and let $\gamma_1\ge\cdots\ge\gamma_n$ be its ordered eigenvalues.
  Conditional block invariance and symmetry under $\Gamma\mapsto-\Gamma$
  give
  \begin{equation}\label{eq:bn}
    \E[\Gamma\mid H]=b_nS_H,
    \qquad b_n=\frac1n\E\left(\sum_{i=1}^m\gamma_i
    -\sum_{i=m+1}^n\gamma_i\right).
  \end{equation}
  We prove $b_n\ge1/\sqrt6$ directly, without an asymptotic eigenvalue law.

  Let $B$ be a fixed $m\times m$ off-diagonal block of $\Gamma$.
  The variational principle for the largest $m$ eigenvalues yields
  \[
    \operatorname{tr}(\Gamma S_H)
    \ge\max_{U\in O(m)}\operatorname{tr}
    \left(\Gamma
      \begin{pmatrix}0&U\\U^T&0
    \end{pmatrix}\right)
    =2\|B\|_*.
  \]
  Here $\|B\|_*$ is the sum of its singular values. H\"older applied
  first to the singular values and then to expectations gives
  \[
    \E\|B\|_F^2\le(\E\|B\|_*)^{2/3}
    \big(\E\operatorname{tr}((BB^T)^2)\big)^{1/3}.
  \]
  The entries of $B$ are independent $N(0,1/(2m))$. Gaussian second
  and fourth moments, or Isserlis' formula~\cite{isserlis}, give
  \[
    \E\|B\|_F^2=m/2,
    \qquad\E\operatorname{tr}((BB^T)^2)=(2m+1)/4.
  \]
  For the latter identity, expand
  $\sum_{i,j,a,b}\E B_{ia}B_{ja}B_{jb}B_{ib}$; the three
  Gaussian pairings contribute $(m^3+m^2+m^3)/(4m^2)$.
  Consequently,
  \begin{equation}\label{eq:bn-lower}
    b_n\ge\frac1m\frac{(m/2)^{3/2}}{((2m+1)/4)^{1/2}}
    =\sqrt{\frac{m}{4m+2}}\ge\frac1{\sqrt6}.
  \end{equation}

  Conditional on the ordered eigenvalues, take independent uniform
  $u,v\in S^{m-1}$ and define
  \[
    A_{\rm eig}=\sum_{i=1}^m\gamma_i u_i^2,
    \qquad B_{\rm eig}=\sum_{i=1}^m\gamma_{m+i}v_i^2,
    \qquad b_{n,k}=2^{-k}\E(A_{\rm eig}-B_{\rm eig})^k.
  \]
  The ordering gives $A_{\rm eig}\ge B_{\rm eig}$, and
  $\E(A_{\rm eig}-B_{\rm eig})=2b_n$. Hence Jensen gives
  $b_{n,k}\ge b_n^k\ge6^{-k/2}$.
  Uniform internal eigenvector rotations show that, for fixed $x,H$,
  \[
    \E[(x^T\Gamma x)^k\mid H]
    =\E\left[\left(\frac{1+t}{2}A_{\rm eig}
    +\frac{1-t}{2}B_{\rm eig}\right)^k\right],
    \qquad t=x^TS_Hx.
  \]
  The coefficient of $t^k$ is $b_{n,k}$.
  The leading coefficient of $K_k$ as a polynomial in $t$ is
  $2^k((n+2)/4)_k/k!$, and
  $\E|Y|^{2k}=2^k(n/2)_k$. Thus
  \[
    \alpha_{kk}=
    \frac{k!n^{k/2}b_{n,k}(n/2)_k}{((n+2)/4)_k}
    \ge k!n^{k/2}6^{-k/2},
  \]
  as required.
\end{proof}

\subsection{Local smoothing and coefficient extraction}

We use the classical Poisson operator on the sphere~\cite{groemer}.
For $0\le\rho<1$, let $\mathsf P_\rho x$ be the first exit point
of Brownian motion from the unit ball, started at $\rho x$.
For a homogeneous harmonic polynomial $h_j$ of degree $j$,
bounded optional stopping gives
\begin{equation}\label{eq:poisson}
  \E[h_j(\mathsf P_\rho x)\mid x]=\rho^j h_j(x).
\end{equation}
To justify the stopping, if $\tau$ is the exit time, the squared-norm
martingale gives $n\E(\tau\wedge t)\le1$; hence $\tau$ is almost
surely finite. The polynomial is bounded on the closed ball.
Rotational invariance makes $\sigma_{\rm u}$ stationary, and the real
diagonal action in \eqref{eq:poisson} makes this operator self-adjoint
on $L^2(\sigma_{\rm u})$. At $\rho=1$, take the identity.

\begin{lemma}[Gaussian bridge]\label{lem-bridge}
  For every deterministic cost-$d$ protocol with output constant on each
  rectangle and every $k\ge1$,
  \begin{equation}\label{eq:bridge}
    |\E_0[\psi K_k]|
    \le\frac{(2e\sqrt6)^k}{k!n^{k/2}}\mathcal B_k(d).
  \end{equation}
\end{lemma}

\begin{proof}
  Alice privately replaces $X$ by $\mathsf P_\rho X$ before executing
  the original protocol. This preserves its cost. The comparison
  protocol need not solve the original promise problem: the definition
  of $\mathcal B_k(d)$ applies to every complete rectangle partition.
  Since $K_j(\cdot,H)$ is a degree-$2j$ harmonic, \eqref{eq:poisson}
  gives it eigenvalue $\rho^{2j}$. By \eqref{eq:triangular}, stationarity,
  and self-adjointness,
  \begin{align*}
    P(z)&:=\E[\psi(\mathsf P_\rho X,H)F_k(Y,G)]\\
    &=\sum_{j=0}^k\alpha_{kj}z^j\E_0[\psi(X,H)K_j(X,H)],
    \qquad z=\rho^2.
  \end{align*}
  Represent the smoothing kernel by an independent private seed. For
  every fixed seed, its pullback is again a complete Gaussian rectangle
  partition with at most $2^d$ cells. Therefore
  $|P(z)|\le\mathcal B_k(d)$ for all $z\in[0,1]$.

  For a polynomial of degree at most $k$, its $k$th forward difference
  at spacing $1/k$ gives
  \[
    \frac{k!}{k^k}[z^k]P(z)
    =\sum_{i=0}^k(-1)^{k-i}\binom ki P(i/k).
  \]
  It follows that
  \begin{equation}\label{eq:extraction}
    |[z^k]P(z)|\le\frac{2^kk^k}{k!}\sup_{z\in[0,1]}|P(z)|
    \le(2e)^k\mathcal B_k(d).
  \end{equation}
  Divide by the leading coefficient in \eqref{eq:top-coefficient}.
\end{proof}

The finitely many smoothed protocols are comparison objects in a
coefficient identity. They are not run in succession and do not
multiply the communication cost.

\begin{lemma}[Crude Gaussian bound]\label{lem-crude}
  For a positive-probability rectangle $R=A\times\mathcal C$ with
  surprises $a,b$,
  \[
    \bigl|\E[F_k\mid R]\bigr|\le\big(C(a+b+k)\big)^{3k/2}.
  \]
  Consequently $\mathcal B_k(d)\le C^k(d+k)^{3k/2}$.
\end{lemma}

\begin{proof}
  This is the standard level bound obtained from hypercontractivity
  \cite{nelson}; we include its normalization.
  If a law $\eta$ has density at most $e^s$ relative to standard Gaussian
  measure, \cref{eq:density,eq:classical-chaos} give, for a unit tensor $U$,
  \[
    |\E_\eta\langle U,\He_r\rangle|
    \le e^{s/u}(u-1)^{r/2}\sqrt{r!},
    \qquad u=\max(2,2s/r).
  \]
  This is at most $(2e(s+r))^{r/2}$, so the same bound holds for
  $\|\E_\eta\He_r\|_{\HS}$ by duality. Conditional independence
  on a rectangle and the norm of $\mathcal J_k$ now give
  \begin{align*}
    \bigl|\E[F_k\mid R]\bigr|
    &\le\bigl\|\E[\He_{2k}(Y)\mid A]\bigr\|_{\HS}\,
    2^{k/2}\bigl\|\E[\He_k^{\GOE}(G)\mid\mathcal C]\bigr\|\\
    &\le(C(a+k))^k(C(b+k))^{k/2}.
  \end{align*}
  Average using \cref{lem-surprise} with moment exponent $3k/2$.
\end{proof}

\subsection{Proof of the communication theorem}

\begin{proof}[Proof of \cref{thm-main}, assuming \cref{thm-cp}]
  Let $\psi$ be the output of a deterministic cost-$d$ protocol, with
  $d\ge1$. Set $J=d$ in \cref{lem-trunc}. By
  \eqref{eq:spectral-tail},
  \[
    \Adv(\psi)\le\frac12\sum_{\substack{k\le d\\k\text{ odd}}}
    |\E_0[\psi K_k]|+
    \frac12\left(\frac{4(d+1)}{n-1}\right)^{(d+1)/2}.
  \]
  For $k\le\sqrt d$, \cref{lem-bridge,thm-cp} bound the $k$th
  correlation by
  \[
    \left(\frac{C(d+k^2)}{\sqrt n}\right)^k
    \le\left(\frac{Cd}{\sqrt n}\right)^k.
  \]
  For $\sqrt d<k\le d$, use \cref{lem-crude} and $k!\ge(k/e)^k$:
  \[
    |\E_0[\psi K_k]|\le
    \left(\frac{C(2d)^{3/2}}{k\sqrt n}\right)^k
    \le\left(\frac{Cd}{\sqrt n}\right)^k.
  \]
  Choose an absolute constant $c_0\le1$ sufficiently small that
  $Cc_0\le1/12$. For $1\le d\le c_0\sqrt n$,
  the geometric sum, including its factor $1/2$, is at most
  $r/(2(1-r))<1/12$, where $r=Cc_0$.
  The last term is $O((d+1)/n)$, uniformly in this range: for large $n$
  its base is at most one and its exponent is at least one.
  It is therefore less than $1/12$ for all sufficiently large even $n$,
  giving $\Adv(\psi)<1/6$.

  Now set $d=\lfloor c_0\sqrt n\rfloor$. A randomized protocol
  of cost at most $d-1$ and distributional success at least
  $2/3$ has a fixing of its coins with at least that success.
  Appending its output bit yields a deterministic cost-$d$ protocol,
  contradicting the preceding bound. Worst-case success implies
  distributional success, so, after decreasing $c_0$ by a factor
  of two to absorb rounding, the asserted lower bound follows.
  Raz's protocol gives the matching upper bound~\cite[Theorem~3.8]{raz};
  see also~\cite[Appendix~A]{montanaro}.
\end{proof}

\section{Cancellation over complete rectangle partitions}\label{sec-cancel}
\subsection{One complex random variable}

Work with the complex variable

\[
  Q=\tfrac12w^T(G+iG')w,\qquad w=Y+iZ,
\]

where $Z\sim\gamma_n$ and $G'\sim\mathrm{GOE}(n)$ are independent of each other and of $(Y,G)$. Under the original Gaussian law, \cref{lem-lift} gives

\[
  \mathbb E Q^{2k}=0,\qquad
  \mathbb E[F_k\mid\mathsf T]=2^k\mathbb E[Q^k\mid\mathsf T].
\]

We will refine the information in $\mathsf T$ to include a real center. The next lemma bounds the conditional $k$th moment using the vanishing $2k$th moment.

\begin{lemma}[A vanishing even moment bounds the center]\label{lem-center}

  Fix an integer $k\ge1$ and $\tau\ge0$. Let $Q$ be a complex random variable, let $\mathcal F$ be a sigma-field, and let $\alpha$ be a real $\mathcal F$-measurable random variable. Assume $Q,\alpha\in L^{2k}$ and $\mathbb E Q^{2k}=0$. If

  \begin{equation}\label{eq:center}
    \kappa_l:=\mathbb E[(Q-\alpha)^l\mid\mathcal F],\qquad
    \|\kappa_l\|_{2k/l}\le(l\tau)^l\quad(1\le l\le2k),
  \end{equation}

  then

  \[
    \|\alpha\|_{2k}\le4ek\tau,\qquad
    \bigl\|\mathbb E[Q^k\mid\mathcal F]\bigr\|_2\le(12k\tau)^k.
  \]

  The absolute value is taken \textbf{after} the conditional expectation in \eqref{eq:center}. We never need a dimension-free bound on $\mathbb E|Q-\alpha|^{2k}$.

\end{lemma}
\begin{proof}
  Put $t=\|\alpha\|_{2k}$. Since $\alpha$ is real and $\mathbb E Q^{2k}=0$, expansion and Hölder give

  \[
    t^{2k}\le\sum_{l=1}^{2k}\binom{2k}l t^{2k-l}\|\kappa_l\|_{2k/l}
    \le\sum_{l=1}^{2k}(2ek\tau)^l t^{2k-l}.
  \]

  Here $\binom{2k}l l^l\le(2ek)^l$. This forces $t\le4ek\tau$: otherwise, dividing by $t^{2k}$ gives a geometric sum less than $1$ on the right. This also covers $\tau=0$.

  Now expand the conditional $k$th moment. Set $\kappa_0=1$ and apply Hölder in $L^2$:

  \[
    \bigl\|\mathbb E[Q^k\mid\mathcal F]\bigr\|_2
    \le\sum_{l=0}^k\binom kl\|\alpha^{k-l}\kappa_l\|_2
    \le\sum_{l=0}^k\binom kl t^{k-l}(k\tau)^l
    =(t+k\tau)^k.
  \]

  Since $t\le4ek\tau$ and $4e+1<12$, this proves the claim.
\end{proof}

\subsection{Flats, bulks, and the two kinds of fluctuation}
\label{sec-components}

A \textbf{component} is a law $\rho=\delta^E_y\otimes(\varphi\gamma_W)$. Here $E\subseteq\mathbb R^n$ is a subspace, $W=E^\perp$, $y\in E$, and $\varphi$ is a probability density on $W$ with respect to $\gamma_W$. Under $\rho$, the coordinates in $E$ equal $y$ (the \emph{flat}), and the coordinates in $W$ have density $\varphi$ (the \emph{bulk}). Write $r_E=\dim E$ and $\sigma=\log\sup\varphi$.

Split $w=w_E+w_W$ with $w_E=y+iZ_E$ and $w_W=Y_W+iZ_W$, and split $Q=\frac12w^T(G+iG')w$ accordingly:

\[
  Q=\underbrace{\tfrac12w_E^T(G+iG')_{EE}w_E}_{Q_E}+\underbrace{\tfrac12w_W^T(G+iG')_{WW}w_W+w_E^T(G+iG')_{EW}w_W}_{V}.
\]

The \textbf{center} is the real number $\alpha:=\mathbb E[Q_E\mid G]=\frac12(y^T\,G_{EE}\,y-\operatorname{tr}G_{EE})$, and the \textbf{flat fluctuation} is $\varepsilon:=Q_E-\alpha$. Given $G$ and $Z_E$, $\varepsilon$ is a function of $G'_{EE}$, and $V$ is a function of $(Y_W,Z_W,G'_{WW},G'_{EW})$, so they are independent. The \textbf{tilted bulk moments} are

\[
  \nu_j(v)(G):=\mathbb E\big[V^j\,\big|\,w_E=v\big]\qquad(v\in\mathbb C^E),
\]

where the notation means: replace $w_E$ by $v$ and average over $(Y_W,Z_W,G'_{WW},G'_{EW})$, with $G$ fixed. This defines a polynomial for every $v\in\mathbb C^E$. Then the moments of $Q$ about $\alpha$ are

\begin{equation}\label{eq:fluctuation}
  \kappa_l(G):=\mathbb E\big[(Q-\alpha)^l\big]=\sum_{j=0}^l\binom lj\,\mathbb E_{Z_E,G'_{EE}}\big[\varepsilon^{l-j}\,\nu_j(w_E)(G)\big].
\end{equation}

The flat fluctuation is controlled by Gaussian hypercontractivity and Gaussian vector moments~\cite{nelson,janson}. The bulk is controlled by its signed moments $\nu_j$: if $\varphi\equiv1$, every $\nu_j$ with $j\ge1$ vanishes, because $w_W$ is then a circularly symmetric complex Gaussian and $\mathbb Ew_W^{\otimes r}=0$ for $r\ge1$. This is why signed moments are used: $|V|$ is typically of order $\dim W$, so bounds through $\mathbb E|V|^j$ would depend on the dimension.

\begin{lemma}[Flat fluctuation]\label{lem-flat}

  For every $p\ge l\ge1$,

  \[
    \Big\|\,\|\varepsilon\|_{L^l(Z_E,G'_{EE})}\,\Big\|_{L^p(G)}
    \le Cl\big(|y|^2+r_E+p\big).
  \]

\end{lemma}
\begin{proof}
  Expand

  \[
    \varepsilon=-\tfrac12\big(Z_E^TG_{EE}Z_E-\operatorname{tr}G_{EE}\big)+iy^TG_{EE}Z_E+\tfrac i2w_E^TG'_{EE}w_E .
  \]

  The first term is a centered quadratic Gaussian chaos, so hypercontractivity bounds its $L^l$ norm by $Cl\|G_{EE}\|_{\rm HS}$. The second is linear Gaussian, with norm at most $C\sqrt l\,|G_{EE}\,y|$. Given $Z_E$, the third is complex Gaussian with second moment $\tfrac12|w_E|^4$. Thus, writing $h_l(G)=\|\varepsilon\|_{L^l(Z_E,G'_{EE})}$,

  \[
    h_l(G)\le C\left(l\|G_{EE}\|_{\rm HS}+\sqrt l\,|G_{EE}\,y|+\sqrt l\big(|y|^2+r_E+l\big)\right).
  \]

  Only Gaussian vector lengths remain. The norm $\|G_{EE}\|_{\rm HS}$ is $\sqrt2$ times the length of a standard Gaussian vector of dimension $r_E(r_E+1)/2$, and $G_{EE}\,y$ is Gaussian with covariance $|y|^2I+yy^T$. Hence

  \[
    \Big\|\|G_{EE}\|_{\rm HS}\Big\|_p\le C(r_E+\sqrt p),\qquad
    \Big\||G_{EE}\,y|\Big\|_p\le C|y|(\sqrt{r_E}+\sqrt p).
  \]

  Take the $L^p(G)$ norm above and use $p\ge l$ and $2uv\le u^2+v^2$.
\end{proof}

\begin{lemma}[Bulk tilted moments]\label{lem-tb}

  Suppose $\varphi\le e^\sigma$ is SG($B_1$). Then for $v\in\mathbb C^E$, $1\le j\le2k$ and $p\ge2k$,

  \[
    \|\nu_j(v)\|_{L^{p/j}(G)}\le\big(Cj\,(\sigma+|v|^2+k^2+p)\big)^j .
  \]

\end{lemma}
\begin{proof}
  Put $D=G_{WE}$. It is a standard Gaussian of the matrix space $\mathbb R^{W\times E}$ with the Frobenius inner product, independent of $G_{WW}$. Let $L_v(D):=Dv$, a complex-linear map with $\|L_v\|\le|v|$. Then

  \[
    V=\tfrac12w_W^T(G+iG')_{WW}w_W+\big\langle(D+iD')v,w_W\big\rangle,\qquad D'=G'_{WE}.
  \]

  Expand $V^j$ by the binomial theorem and average over $G'_{WW}$, $D'$ and $w_W$. These are independent, so each factor becomes a Wick tensor: $\mathbb E(G+iG')^{\otimes m}_{WW}=\mathrm{He}_m^{\rm GOE}(G_{WW})$, $\mathbb E(D+iD')^{\otimes c}=\mathrm{He}_c(D)$, and $\mathbb Ew_W^{\otimes(2m+c)}=M_{2m+c}(\varphi)$, the bulk's Wick tensor (\cref{sec-bulk}). Therefore

  \begin{equation}\label{eq:tilt}
    \nu_j(v)=\sum_{c=0}^j\binom jc2^{c-j}\,U^{L_v}_{j-c,c}(\varphi),
  \end{equation}

  where $U^L_{m,c}(\varphi):=\langle M_{2m+c}(\varphi),\mathcal J_m\mathrm{He}^{\rm GOE}_m(G_{WW})\otimes L^{\otimes c}\mathrm{He}_c(D)\rangle$ is the mixed chaos of \cref{sec-bulk}. When $v=0$ or $\dim E=0$, every term with $c\ge1$ vanishes.

  Apply \cref{prop-msb} on $W$ with $q=p/j\ge1$. Each summand satisfies

  \[
    \|U^{L_v}_{j-c,c}\|_{L^{p/j}}\le|v|^c(Bj)^j(\sigma+12j^2+p+j)^{j-c/2}\le(Bj)^jX^j,\qquad X:=\sigma+48k^2+p+2k+|v|^2.
  \]

  The coefficients in \eqref{eq:tilt} sum to $(3/2)^j$.
\end{proof}

\subsection{Proof of the complete-partition estimate}
\label{sec-cp-proof}
\begin{proof}[Proof of \cref{thm-cp}]

  Fix a complete partition $\mathsf T$ into at most $2^d$ rectangles and an integer $k\ge1$. We apply \cref{lem-center} on the probability space obtained by sampling the entire partition.

  \textbf{Sample the rectangles and components.} First choose $R=A\times\mathcal C$ with probability $\Pr(R)$, and write $a,b$ for its surprises. By \cref{thm-decomp},

  \[
    \mathcal L(Y\mid A)=\int\rho_\omega\,d\pi_R(\omega),\qquad
    \rho_\omega=\delta^{E_\omega}_{y_\omega}\otimes(\varphi_\omega\gamma_{E_\omega^\perp}),
  \]

  where each $\varphi_\omega$ is SG($B_1$) and

  \[
    c_\omega=\sigma_\omega+|y_\omega|^2+\dim E_\omega,\qquad
    \|c_\omega\|_{L^u(\pi_R)}\le4a+8u.
  \]

  Given $R$, choose $\omega\sim\pi_R$ and $G\mid\mathcal C$ independently, then sample $Y\sim\rho_\omega$. Completeness makes the marginal law of $(Y,G)$ exactly $\gamma_n\otimes\mathrm{GOE}(n)$, with $R=\mathsf T(Y,G)$ almost surely. Add independent $Z,G'$.

  Let $\mathcal F$ record $(R,\omega,G)$, and let $\alpha$ be the real center from \cref{sec-components}. Thus $\mathsf T$ is included in $\mathcal F$, and $\mathbb E Q^{2k}=0$. It remains to prove \eqref{eq:center} with $\tau=C(d+k^2)$.

  \textbf{Bound the centered moments.} Fix $R$ and a component, suppress $\omega$, and put

  \[
    p=b+2k,\qquad X=c_\omega+k^2+p,\qquad
    h_l(G)=\|\varepsilon\|_{L^l(Z_E,G'_{EE})}.
  \]

  For $l\le2k$, \cref{lem-flat,lem-tb} and $\||Z_E|^2\|_p\le\dim E+2p$ give

  \[
    \|h_l\|_{L^p(G)}\le ClX,\qquad
    \|\nu_j(w_E)\|_{L^{p/j}(Z_E,G)}\le(CjX)^j\quad(1\le j\le l).
  \]

  Apply Hölder first to the auxiliary variables in \eqref{eq:fluctuation}, then to $G$. Since $p\ge l$, the result is

  \begin{equation}\label{eq:fluctuation-bound}
    \|\kappa_l\|_{L^{p/l}(G)}
    \le\sum_{j=0}^l\binom lj
    \|h_l\|_{L^p(G)}^{l-j}\,
    \|\nu_j(w_E)\|_{L^{p/j}(Z_E,G)}
    \le(ClX)^l.
  \end{equation}

  At $j=0$ the bulk factor is $1$; at $j=l$ the flat factor is $1$. The mixed-norm step uses $\|\|f\|_{L^{l/j}(Z_E)}\|_{L^{p/j}(G)}\le\|f\|_{L^{p/j}(Z_E,G)}$.

  The density of $G\mid\mathcal C$ is at most $e^b$. Apply \eqref{eq:density} to \eqref{eq:fluctuation-bound}, from exponent $p/l$ to $2k/l$. Its cost is $e^{bl/p}\le e^l$, so

  \[
    \|\kappa_l^\omega\|_{L^{2k/l}(G\mid\mathcal C)}
    \le\big(Cl(c_\omega+k^2+b+2k)\big)^l.
  \]

  \textbf{Average once over the whole partition.} Taking the $L^{2k/l}$ norm over $(R,\omega)$ as well gives

  \[
    \|\kappa_l\|_{2k/l}
    \le(Cl)^l\|c_\omega+b+k^2+2k\|_{2k}^l
    \le(Cl)^l\|4a+b+k^2+18k\|_{L^{2k}(R)}^l
    \le\big(Cl(d+k^2)\big)^l.
  \]

  The middle inequality uses the component cost bound with $u=2k$; the last uses \cref{lem-surprise}. This proves \eqref{eq:center}.

  Since $\mathbb E Q^{2k}=0$, \cref{lem-center} now bounds the conditional $k$th moment. The tower property and $L^2$ contraction give

  \[
    \mathbb E\bigl|\mathbb E[F_k\mid\mathsf T]\bigr|
    =2^k\bigl\|\mathbb E[Q^k\mid\mathsf T]\bigr\|_1
    \le2^k\bigl\|\mathbb E[Q^k\mid\mathcal F]\bigr\|_2
    \le\big(Ck(d+k^2)\big)^k
    \le C^kk!(d+k^2)^k.
  \]

  This is \eqref{eq:cp}. The decomposition is measurable. Its cost bound, Gaussian moments, and the finite number of positive-probability rectangles ensure $\alpha\in L^{2k}$. Also, $Q\in L^{2k}$ because its marginal law is the original Gaussian law.
\end{proof}
\section{Flats and sub-Gaussian components}\label{sec-decomp}
The affine-slicing construction below is related to the Gaussian cleanup
in Girish--Sinha--Tal--Wu~\cite[Sections~2 and~5]{gstw}.
Their cleanup controls directional covariance during a protocol. Here
the stopping condition tests all one-dimensional tails, and a single
potential controls the residual density, the fixed coordinates, and
the dimension simultaneously. Compare also Eldan's
entropy-efficient decompositions with covariance control~\cite{eldan};
our conclusion is an exact mixture of affine slices with an exponential
bound on this combined cost.

We use the fixed constant $B_1=5$. The decomposition will stop when all one-dimensional tails satisfy the following test.

\begin{lemma}[A tail test for sub-Gaussianity]\label{lem-inverse}

  A law is SG($5$) if, for every unit vector $v$ and every $x>0$,

  \begin{equation}\label{eq:regular}
    \Pr(|\langle v,Y\rangle|\ge x)\le e^{2-x^2/6}.
  \end{equation}

\end{lemma}
\begin{proof}
  For $V=\langle v,Y\rangle$, \eqref{eq:regular} says that $V^2$ is stochastically dominated by $12+6E$, where $E$ is exponential with mean one. For $q\ge2$,

  \[
    \|V\|_q^2\le12+6\|E\|_{q/2}\le12+3q\le9q.
  \]

  For $1\le q<2$, use $\|V\|_q\le\|V\|_2\le\sqrt{18}<5\sqrt q$.
\end{proof}

\begin{theorem}[Decomposition into flats and sub-Gaussian bulks]\label{thm-decomp}

  Let $\rho\le e^a$ be a probability density on a Euclidean space $\Omega$ relative to $\gamma_\Omega$. There is a measurable mixture

  \[
    \rho\gamma_\Omega=\int\delta^{E_\omega}_{y_\omega}\otimes(\varphi_\omega\gamma_{E_\omega^\perp})\,d\pi(\omega),
  \]

  with every $\varphi_\omega$ SG($5$). Its cost $c_\omega=\log\sup\varphi_\omega+|y_\omega|^2+\dim E_\omega$ satisfies

  \begin{equation}\label{eq:component-cost}
    \mathbb E_\pi e^{c_\omega/8}\le e^{a/2},\qquad
    \pi(c_\omega>4a+8t)\le e^{-t},\qquad
    \|c_\omega\|_u\le4a+8u \quad(t\ge0,\ u\ge1).
  \end{equation}

\end{theorem}
\begin{proof}
  \textbf{Split a failing tail.} A state $(E,y,\varphi)$ represents the law $\delta_y^E\otimes\varphi\gamma_{E^\perp}$. Start at $(\{0\},0,\rho)$ and stop when \eqref{eq:regular} holds. Otherwise choose a violating pair $(v,x)$ in $E^\perp$. Let $f$ be the marginal density of $\langle v,Y\rangle$ relative to $\gamma_1$, and put

  \[
    B=\{|t|\ge x\},\qquad p=\int_B f\,d\gamma_1>e^{2-x^2/6}.
  \]

  With probability $1-p$, restrict the bulk to $B^c$ and normalize. With probability measure $\mathbf1_B(t)f(t)\,d\gamma_1(t)$, fix the coordinate at $t$: replace $(E,y)$ by $(E+\mathbb Rv,y+tv)$ and use the remaining density $\varphi_t(z)=\varphi(tv+z)/f(t)$. These transitions preserve the represented law by disintegration. If $\sigma=\log\sup\varphi$, their log-density bounds are

  \begin{equation}\label{eq:fiber-density}
    \sigma_{\rm off}\le\sigma-\log(1-p),\qquad
    \sigma_t\le\sigma-\log f(t).
  \end{equation}

  The off-tail child is omitted when $p=1$; zero-density fibers have zero probability.

  \textbf{One potential controls the entire cost.} Put

  \[
    \mathcal V(E,y,\varphi)=\exp\!\left(\frac\sigma2+\frac{|y|^2+\dim E}{8}\right).
  \]

  A violating tail has $x^2>12$. Moreover,

  \[
    I_x:=\int_B e^{t^2/4}\,d\gamma_1(t)
    =\sqrt2\,\gamma_1(|g|\ge x/\sqrt2)
    \le\sqrt2e^{-x^2/4},\qquad
    \frac{I_x}{p}\le\sqrt2e^{-3}.
  \]

  Using \eqref{eq:fiber-density}, orthogonality of the new flat coordinate, and Cauchy--Schwarz, we obtain

  \[
    \begin{aligned}
      \frac{\mathbb E[\mathcal V_{\rm next}\mid E,y,\varphi]}{\mathcal V(E,y,\varphi)}
      &\le\sqrt{1-p}+e^{1/8}\int_B\sqrt{f(t)}e^{t^2/8}\,d\gamma_1(t)\\
      &\le\sqrt{1-p}+e^{1/8}\sqrt{pI_x}
      \le\sqrt{1-p}+\frac p3\le1-\frac p6.
    \end{aligned}
  \]

  Here $e^{1/8}\sqrt{I_x/p}\le2^{1/4}e^{-11/8}<1/3$. Thus $\mathcal V$ is a nonnegative supermartingale, held constant after stopping.

  \textbf{The same potential proves termination.} Since $f\le e^\sigma$, the inequalities $p>e^{2-x^2/6}$ and $p\le e^{\sigma-x^2/2}$ imply $p\ge e^{-\sigma/2}$. Hence $p\mathcal V\ge1$ at every nonterminal state: each split decreases the conditional expectation of $\mathcal V$ by at least $1/6$. If $\tau$ is the number of splits, telescoping gives

  \[
    \tfrac16\mathbb E\min(\tau,N)\le\mathcal V_{\rm root}-\mathbb E\mathcal V_N\le\mathcal V_{\rm root}.
  \]

  Letting $N\to\infty$ proves $\mathbb E\tau\le6e^{a/2}<\infty$, so the process stops almost surely.

  Let $\pi$ be its terminal law. The finite-step disintegration identity passes to the limit, since the probability of not yet stopping tends to zero. This proves the mixture identity. At the terminal state, $e^{c_\omega/8}\le\mathcal V_\omega$. Fatou's lemma gives

  \[
    \mathbb E_\pi e^{c_\omega/8}\le\mathbb E_\pi\mathcal V_\omega\le\mathcal V_{\rm root}\le e^{a/2}.
  \]

  Markov's inequality gives the tail in \eqref{eq:component-cost}; $(c_\omega-4a)^+/8$ is dominated by an exponential variable, giving the moment bound.

  \textbf{Measurability.} Use Borel versions of the densities. Fix a countable dense set of ambient directions, project them onto $E^\perp$ and normalize the nonzero projections. Test these directions and rational $x>0$, and choose the first strict violation of \eqref{eq:regular}. This detects every failure: for a Gaussian density the tail probability is continuous in $(v,x)$, because each boundary hyperplane has measure zero. The transitions are Borel kernels on states described by finitely many directions, offsets and cuts. Their almost surely finite terminal state is therefore measurable.
\end{proof}
\section{Gaussian chaos estimates}\label{sec-chaos}
The partition-norm framework comes from Lata{\l}a~\cite{latala};
Lehec~\cite{lehec-chaos} gives a different proof of his estimates.
We combine the log-Sobolev moment method with a Gaussian-slot inequality
adapted from Aden-Ali~\cite{aden-ali}. The constants are kept explicit
because the degree grows with the communication cost.
Iterating a derivative estimate is also central to Adamczak--Wolff
\cite{adamczak-wolff}. Here we stop the iteration halfway and retain
its Hilbert--Schmidt contribution separately.

Let $A\in H_1\otimes\cdots\otimes H_r$, let $G_1,\dots,G_r$ be independent standard Gaussians in $H_1,\dots,H_r$, and put $X=A(G_1,\dots,G_r)$. The goal is \cref{lem-split}, a moment bound whose Hilbert--Schmidt term costs only $q^{\lfloor r/2\rfloor/2}$ instead of the hypercontractive $q^{r/2}$.

The idea is to stop a derivative expansion halfway. Its $L^2$ terms are bounded by the Frobenius norm. At the stopping point, fewer than $r/2$ Gaussian arguments remain; each costs at most $2r$ through the slot estimate below.

For $T\subseteq[r]$ put

\[
  \varphi_T:=\sup\big\lbrace A(x_T,G_{T^c}):\ x_a\in H_a,\ |x_a|\le1\ (a\in T)\big\rbrace,
\]

with $\varphi_\emptyset:=|X|$. (In this section, $T$ is a set of slots, and $P,Q$ denote partitions.) So $\varphi_T$ is the injective norm, over the slots in $T$, of the partially evaluated tensor. Each $\varphi_T$ is a countable supremum of continuous functions, locally Lipschitz, and bounded by $\|A\|_{\rm HS}\prod_{c\notin T}|G_c|$.

\begin{lemma}[Moment growth of partial tensors; cf.\ {\cite{aida-stroock}}]\label{lem-lsi}

  Put $H_T(s)=\|\varphi_T\|_s^2$. For $s\ge2$,

  \begin{equation}\label{eq:growth}
    H_T(2)\le\|A\|_{\rm HS}^2,\qquad
    H_T'(s)\le\sum_{b\notin T}H_{T\cup\lbrace b\rbrace}(s).
  \end{equation}

\end{lemma}
\begin{proof}
  The injective norm is at most the Hilbert--Schmidt norm, and Gaussian contraction preserves the expected square of the latter:

  \[
    H_T(2)\le\mathbb E\|A(\cdot_T,G_{T^c})\|_{\rm HS}^2=\|A\|_{\rm HS}^2.
  \]

  For a nonnegative locally Lipschitz $f$ of polynomial growth, log-Sobolev applied to $f^{s/2}$ and then Hölder give

  \[
    \frac d{ds}\|f\|_s^2
    =\frac{2\|f\|_s^{2-s}}{s^2}\,\mathrm{Ent}(f^s)
    \le\|\,|\nabla f|\,\|_s^2.
  \]

  This is the classical moment-growth argument of Aida--Stroock~\cite{aida-stroock}; see also the entropy method in~\cite{bblm}. The underlying Gaussian log-Sobolev inequality is due to Gross~\cite{gross}.

  With the other Gaussians fixed, $G_b\mapsto\varphi_T$ is a supremum of linear functions whose slopes have norm at most $\varphi_{T\cup\lbrace b\rbrace}$. Thus $|\nabla_{G_b}\varphi_T|\le\varphi_{T\cup\lbrace b\rbrace}$ almost everywhere. Apply the moment-growth inequality to $\varphi_T$ and use Minkowski in $L^{s/2}$.
\end{proof}

\begin{lemma}[One Gaussian slot; adapted from {\cite{aden-ali}}]\label{lem-aa}

  Let $D\in H_0\otimes\cdots\otimes H_b$ with $b\ge1$, let $g$ be standard Gaussian in $H_0$, and put $Z(g)=\|D(g,\cdot)\|_{\rm inj}$. For $I\subseteq[b]$, let $Q_I$ keep the slots in $I$ as singletons and join all remaining slots, including $0$, into one block. For every $p\ge2$,

  \begin{equation}\label{eq:slot}
    p^b\|Z\|_p^2\le(b^2+1)\sum_{I\subseteq[b]}b^{|I|-b}p^{|I|+1}\|D\|_{Q_I}^2.
  \end{equation}

\end{lemma}
\begin{proof}
  This specializes the Gaussian-shift argument of Aden-Ali~\cite[proof of Theorem~1.2]{aden-ali} directly to moments. For comparison, \cite[Corollary~4.1]{aden-ali} obtains a slot moment bound from the expectation estimate and Gaussian concentration. We obtain the tail bound within the shift argument and integrate it. Assume $D\ne0$, and put

  \[
    \beta=bp,\qquad V=\sum_{I\subseteq[b]}\beta^{|I|-b}\|D\|_{Q_I}^2.
  \]

  Let $\pi$ be the law of $\theta=\beta^{-1/2}(\eta_1,\dots,\eta_b)$ with independent standard Gaussians $\eta_i$. For unit-ball points $x=(x_1,\dots,x_b)$, let $\rho_x$ be the law of $x+\theta$. Its mean log density ratio to $\pi$ is $\frac\beta2\sum_i|x_i|^2\le\beta b/2$. Multilinearity gives $\mathbb E_{\rho_x}D(g,\theta)=D(g,x)$. Orthogonality of the Gaussian terms gives, with $v(\theta)=\|D(\cdot,\theta)\|_{\rm HS}^2$,

  \[
    \mathbb E_{\rho_x}v(\theta)
    =\sum_{I\subseteq[b]}\beta^{|I|-b}\|D(\cdot_0,x_I,\cdot_{I^c})\|_{\rm HS}^2\le V.
  \]

  For $\lambda>0$, define

  \[
    W_\lambda(g)=\mathbb E_\pi\exp\!\left(\lambda D(g,\theta)-\frac{\lambda^2}2v(\theta)\right).
  \]

  The Gaussian moment-generating function gives $\mathbb E_gW_\lambda=1$. Jensen under $\rho_x$, followed by the supremum over $x$, gives

  \[
    \lambda Z(g)\le\log W_\lambda(g)+\frac{\lambda^2V}2+\frac{\beta b}2.
  \]

  For each $t\ge0$, Markov's inequality and the choice $\lambda=\sqrt{(\beta b+2t)/V}$ therefore give

  \[
    \Pr\!\left(Z>\sqrt{V(\beta b+2t)}\right)\le e^{-t}.
  \]

  So $Z^2$ is dominated by $V(\beta b+2\mathsf E)$ with $\mathsf E$ exponential of mean one. Since $p/2\ge1$,

  \[
    \|Z\|_p^2\le V\big(\beta b+2\|\mathsf E\|_{p/2}\big)\le V(\beta b+p).
  \]

  Substitute $\beta=bp$ to obtain \eqref{eq:slot}.
\end{proof}

\begin{lemma}[Completing the remaining slots]\label{lem-r}

  For every $T\subseteq[r]$ with $|T|=k\ge1$ and every $p\ge2$,

  \[
    p^{k/2}\|\varphi_T\|_p\le(2r)^{r-k}\max_{Q\in\mathcal P(r)}p^{|Q|/2}\|A\|_Q .
  \]

\end{lemma}
\begin{proof}
  For a partition $P$ of $S\subseteq[r]$, put $Y_P=\|A(\cdot_S,G_{S^c})\|_P$ and $N_P=p^{|P|/2}\|Y_P\|_p$. Fix $c\notin S$ and put $b=|P|$. For $I\subseteq P$, let $Q(P,I)$ keep the blocks in $I$ and merge all other blocks with $c$.

  Condition on all Gaussian slots except $c$, flatten the blocks of $P$, and apply \eqref{eq:slot}. Taking $L^{p/2}$ norms in the remaining Gaussians yields

  \[
    N_P^2\le(b^2+1)\sum_{I\subseteq P}b^{|I|-b}N_{Q(P,I)}^2
    \le(b^2+1)(1+1/b)^b\max_I N_{Q(P,I)}^2
    \le4b^2\max_I N_{Q(P,I)}^2.
  \]

  Here the powers of $p$ agree because $|Q(P,I)|=|I|+1$. Start from the singleton partition of $T$, for which $N_P=p^{k/2}\|\varphi_T\|_p$. Add the other $r-k$ slots, charging at most $2r$ each time because $b\le r$. At the end, $N_Q=p^{|Q|/2}\|A\|_Q$.
\end{proof}

\begin{theorem}[Split moment bound]\label{lem-split}

  For every real tensor $A$ of order $r\ge1$ and every $q\ge2$,

  \[
    \|A(G_1,\dots,G_r)\|_q\le4^r\Big(q^{\lfloor r/2\rfloor/2}\|A\|_{\rm HS}+r^{(r-1)/2}\max_{Q\in\mathcal P(r)}q^{|Q|/2}\|A\|_Q\Big).
  \]

\end{theorem}
\begin{proof}
  Put $h=\lfloor r/2\rfloor$, $k=h+1$ and $N^*=\max_Qq^{|Q|/2}\|A\|_Q$. Integrate \eqref{eq:growth}, starting with $H_\emptyset(q)=\|X\|_q^2$, and stop after $k$ steps. Each set of $l$ slots occurs in $l!$ orders, which cancels the $1/l!$ from the nested integrals. Using $H_T(2)\le\|A\|_{\rm HS}^2$ and $H_T(s)\le H_T(q)$ for $s\le q$ gives

  \begin{equation}\label{eq:half}
    \|X\|_q^2
    \le\|A\|_{\rm HS}^2\sum_{l=0}^{h}\binom rl(q-2)^l
    +(q-2)^k\sum_{|T|=k}\|\varphi_T\|_q^2.
  \end{equation}

  The first term is at most $2^rq^h\|A\|_{\rm HS}^2$. For every term in the second sum, \cref{lem-r} gives

  \[
    q^{k/2}\|\varphi_T\|_q\le(2r)^{r-k}N^*.
  \]

  There are at most $2^r$ such sets and $2(r-k)\le r-1$. Thus the second term in \eqref{eq:half} is at most $4^rr^{r-1}(N^*)^2$. Take square roots.
\end{proof}

The two terms play different roles in \cref{sec-bulk}. The Hilbert--Schmidt term is paired with the Frobenius bound, and the other with the partition norms.

The following Gaussian decoupling identities are classical; see the general framework in~\cite{janson}. We give the conditional-Gaussian calculation to specify the one-sided constants.

\begin{lemma}[Hermite decoupling]\label{lem-dec}
  \leavevmode
  \begin{enumerate}[label=(\alph*)]
    \item Let $X_1,\dots,X_r$ be independent standard Gaussians in $\mathcal H$ and $\bar X=r^{-1/2}\sum_iX_i$. Then $\mathbb E[X_1\otimes\cdots\otimes X_r\mid\bar X]=r^{-r/2}\mathrm{He}_r(\bar X)$. Hence $\|\langle A,\mathrm{He}_r(\bar X)\rangle\|_q\le r^{r/2}\|A(X_1,\dots,X_r)\|_q$ for $q\ge1$.
    \item Let $W_i=Y_i+iZ_i$ with all $Y_i,Z_i$ independent standard real Gaussians, and $\bar W=r^{-1/2}\sum_iW_i$. Then $\mathbb E[W_1\otimes\cdots\otimes W_r\mid\bar W]=r^{-r/2}\bar W^{\otimes r}$. Hence $\|\langle A,\bar W^{\otimes r}\rangle\|_q\le r^{r/2}\|A(W_1,\dots,W_r)\|_q$.
  \end{enumerate}
\end{lemma}
\begin{proof}
  Write $X_i=\bar X/\sqrt r+\xi_i$. The residuals are independent of $\bar X$, with $\mathbb E\xi_i\otimes\xi_{i'}=(\delta_{ii'}-1/r)\mathrm{Id}$. By Isserlis\textquotesingle{} formula~\cite{isserlis}, a term with unpaired set $S$ and pairing $\pi$ of the rest has weight $r^{-|S|/2}(-1/r)^{|\pi|}=r^{-r/2}(-1)^{|\pi|}$; the sum is $r^{-r/2}\mathrm{He}_r(\bar X)$. In (b), the bilinear covariances of the complex residuals vanish, so only the term with no residuals survives. Conditional Jensen gives the norm bounds.
\end{proof}

\begin{corollary}[Complex split bound]\label{lem-cplx}

  Let $T\in(\mathbb C^N)^{\otimes k}$ with $k\ge2$, and $W=Y+iZ$ with $Y,Z$ independent standard Gaussians. For $s\ge2$,

  \[
    \Lambda_s(T):=\|\langle T,W^{\otimes k}\rangle\|_s\le8^kk^{k/2}\Big(s^{k/4}\|T\|_{\rm HS}+k^{(k-1)/2}\max_{J\in\mathcal P(k)}s^{|J|/2}\|T\|_J\Big).
  \]

\end{corollary}
\begin{proof}
  \cref{lem-dec}(b) decouples the slots at cost $k^{k/2}$. Write $W_i=\iota(Y_i,Z_i)$ with $\iota(y,z)=y+iz$, and apply \cref{lem-split} to the real and imaginary parts of $T\circ\iota^{\otimes k}$. A real unit test tensor on a block $g$ maps to a complex tensor of norm at most $2^{|g|/2}$, so every partition norm grows by at most $2^{k/2}$. Adding the real and imaginary parts costs $2$, and $2^{k/2+1}\le2^k$ for $k\ge2$.
\end{proof}
\section{Polynomial moments of sub-Gaussian components}\label{sec-bulk}
We now prove \cref{prop-msb}, the bulk estimate used in
\cref{lem-tb}. The argument uses the split Gaussian estimate of
\cref{lem-split,lem-cplx}, with its relation to earlier chaos estimates
explained in \cref{sec-chaos}. Here the law need only have sub-Gaussian
linear marginals and a bounded density relative to Gaussian measure;
its coordinates may be dependent. In particular, we do not assume the
functional inequalities on the law itself that underlie the
polynomial concentration theorems of Adamczak--Wolff~\cite{adamczak-wolff}.
One induction in weighted polynomial moments gives both the partition
norms and the stronger Frobenius bound.

\textbf{Setting.} $W$ is a Euclidean space of dimension $N$; all bounds are independent of $N$. $\varphi\le e^\sigma$ ($\sigma\ge0$) is an SG($B_1$) probability density on $W$ relative to $\gamma=\gamma_W$. Let $Y\sim\varphi\gamma$ and $Z\sim\gamma$ be independent and $w:=Y+iZ$. The \textbf{Wick tensor} of order $r$ is

\[
  M_r(\varphi):=\mathbb Ew^{\otimes r}=\mathbb E_\varphi\mathrm{He}_r(Y),
\]

a real symmetric tensor. It is zero for every $r\ge1$ when $\varphi\equiv1$.

The complex representation gives products for free. For a partition of $[r]$ into blocks $\beta$ and test tensors $y^\beta$,

\[
  \Big\langle M_r,\bigotimes_\beta y^\beta\Big\rangle=\mathbb E\prod_\beta a_{y^\beta}(w),\qquad a_x(w):=\langle x,w^{\otimes k}\rangle\ \ (x\in(\mathbb C^N)^{\otimes k}).
\]

So partition norms of $M_r$ reduce to moments of the polynomials $a_x$. Put

\[
  \Theta_k(p):=\sup\big\lbrace\|a_x\|_{L^p}:\ \|x\|\le1\big\rbrace\qquad(p\ge1),\qquad R:=\max\lbrace2(B_1+1),128\rbrace,
\]

the norm being over $(Y,Z)$.

\begin{proposition}[Polynomial moment bounds]\label{prop-boot}

  Let $\mu\ge1$ and $\Sigma=\sigma+2\mu^2+\mu$. For every integer $1\le k\le\mu$,

  \[
    \Theta_k(\mu/k)\le R^k(\mu/k)^{1/2}\Sigma^{(k-1)/2}.
  \]

  At $p=1$, the same argument also gives

  \begin{equation}\label{eq:l1}
    \Theta_k(1)\le R^k k^{k/2}(\sigma+2k^2+k)^{k/4}.
  \end{equation}

\end{proposition}
\begin{proof}
  Keep $kp=\mu$ fixed and induct on $k$. This lets every block of a partition use the same induction hypothesis. For $k=1$, separating real and imaginary parts of a complex unit vector gives $\Theta_1(\mu)\le2(B_1+1)\sqrt\mu$, which also implies \eqref{eq:l1} at $k=1$.

  \textbf{Use duality.} For $k\ge2$ put $p=\mu/k$ and $s=\Sigma/k$. Then $s\ge p$, $s\ge2k$, and \eqref{eq:density} costs at most $e^{\sigma/s}\le e^k$. By $L^p$ duality,

  \[
    \Theta_k(p)=\sup_{\|\zeta\|_{p'}\le1}\big\|\mathbb E[\zeta w^{\otimes k}]\big\|_{\rm HS},\qquad p'=\frac p{p-1},
  \]

  where $p'=\infty$ when $p=1$, and $\zeta$ is any complex function of $(Y,Z)$. Fix such a $\zeta$ and write $A=\mathbb E[\zeta w^{\otimes k}]$ and $F=\|A\|_{\rm HS}$. Hölder and \eqref{eq:density} give

  \[
    F^2=\mathbb E[\zeta\langle\bar A,w^{\otimes k}\rangle]\le\|\langle\bar A,w^{\otimes k}\rangle\|_p\le e^k\Lambda_s(\bar A),
  \]

  where $\Lambda_s$ is the Gaussian moment of \cref{lem-cplx}.

  \textbf{Bound the partition norms.} The one-block norm of $A$ is $F$. If $J$ has $l\ge2$ blocks of sizes $k_\beta<k$, generalized Hölder with exponents $p'$ and $\mu/k_\beta$ gives

  \[
    \|A\|_J\le\prod_{\beta\in J}\Theta_{k_\beta}(\mu/k_\beta).
  \]

  The induction hypothesis therefore gives

  \[
    s^{l/2}\|A\|_J\le R^k\Sigma^{k/2}p^{l/2}\le R^k(\Sigma p)^{k/2}.
  \]

  Complex conjugation preserves these norms, so the same bounds hold for $\bar A$.

  \textbf{One quadratic gives both bounds.} Insert these estimates into \cref{lem-cplx}, separating its Hilbert--Schmidt term, the one-block member $s^{1/2}F$ of its maximum, and the other partitions. With $C_*=8e$,

  \begin{equation}\label{eq:bootstrap}
    \begin{aligned}
      F^2&\le(A_1+A_2)F+A_3,\\
      A_1&=C_*^kk^{k/4}\Sigma^{k/4},\qquad
      A_2=C_*^kk^{k-1}\Sigma^{1/2},\\
      A_3&=C_*^kk^{k-1/2}R^k(\Sigma p)^{k/2}.
    \end{aligned}
  \end{equation}

  Let $H=R^k\sqrt p\,\Sigma^{(k-1)/2}$ and $\eta=C_*/R<1/4$. Since $\Sigma\ge k^2$, $\Sigma/p\ge k^2$ and $p\ge1$,

  \[
    A_1,A_2\le k\eta^kH,\qquad A_3\le k^{3/2}\eta^kH^2.
  \]

  For $k\ge2$, $k^{3/2}\le2^k$, so $A_1+A_2\le H/2$ and $A_3\le H^2/4$. By \eqref{eq:quadratic}, $F\le H$. Taking the supremum over $\zeta$ proves the induction.

  For \eqref{eq:l1}, take $\mu=k$, so $p=1$ and $\Sigma=\sigma+2k^2+k$, and put $K=R^kk^{k/2}\Sigma^{k/4}$. Since $\Sigma\ge k^2$, the same coefficients in \eqref{eq:bootstrap} satisfy

  \[
    A_1,A_2\le\eta^kK,\qquad A_3\le\eta^kK^2.
  \]

  Thus \eqref{eq:quadratic} gives $F\le(2\eta^k+\eta^{k/2})K\le K$. Taking the supremum over $\zeta$ proves \eqref{eq:l1}.
\end{proof}

The inflation $2\mu^2$ in $\Sigma$ is exactly what makes the closing step work: it gives $s\ge2k$ and $\Sigma,\Sigma/p\ge k^2$, with no logarithmic margin.

\begin{corollary}[Partition norms]\label{cor-reg}

  For every $r\ge1$ and every partition $J$ of $[r]$ with $l$ blocks,

  \[
    \|M_r(\varphi)\|_J\le R^r\,r^{l/2}\,\Sigma_r^{(r-l)/2},\qquad\Sigma_r:=\sigma+2r^2+r .
  \]

\end{corollary}
\begin{proof}
  Write $\langle M_r,\otimes_\beta y^\beta\rangle=\mathbb E\prod_\beta a_{y^\beta}(w)$. Apply Hölder with exponents $r/k_\beta$, then \cref{prop-boot} with $\mu=r$.
\end{proof}

So each block gains a factor $\sqrt{r/\Sigma_r}$ over the trivial scale $\Sigma_r^{r/2}$. For the one-block norm, \eqref{eq:l1} gives the $\Sigma_r^{r/4}$ dependence needed below.

\begin{corollary}[Frobenius bound]\label{lem-frob}

  \[
    \|M_r(\varphi)\|_F\le R^r\,r^{r/2}\,\Sigma_r^{r/4}.
  \]

\end{corollary}
\begin{proof}
  $\|M_r\|_F=\sup_{\|x\|\le1}|\mathbb E a_x(w)|\le\Theta_r(1)$; apply \eqref{eq:l1}.
\end{proof}

\textbf{Mixed chaos.} Let $G\sim\mathrm{GOE}(N)$, let $g$ be an independent standard Gaussian of a Euclidean space $\mathcal H$, and let $L:\mathcal H_{\mathbb C}\to\mathbb C^N$ be complex-linear with $\lambda=\|L\|$. For integers $m,c\ge0$ put $r=2m+c$ (the tensor order) and $j=m+c$ (the Gaussian degree), so $j\le r\le2j$. Define

\[
  U^L_{m,c}(\varphi)=\big\langle M_r(\varphi),\ \mathcal J_m\mathrm{He}^{\rm GOE}_m(G)\otimes L^{\otimes c}\mathrm{He}^{\mathcal H}_c(g)\big\rangle,
\]

a pure chaos of degree $j$ in $(G,g)$. In \cref{lem-tb}, $G=G_{WW}$, $g=D$ and $L=L_v$.

\begin{lemma}[Decoupled form of the mixed chaos]\label{lem-mixed}

  Let $A_L(X_1,\dots,X_m,v_1,\dots,v_c)=\langle M_r,\mathcal J_m(X_1\otimes\cdots\otimes X_m)\otimes Lv_1\otimes\cdots\otimes Lv_c\rangle$. Then:

  \begin{enumerate}[label=(\alph*)]
    \item $\|U^L_{m,c}\|_q\le j^{j/2}\|A_L(G_1,\dots,G_m,g_1,\dots,g_c)\|_q$ for $q\ge1$, with independent Gaussian inputs;
    \item $\|A_L\|_J\le2^{m/2}\lambda^c\|M_r\|_{\hat J}$ for every partition $J$ of the $j$ slots, where $\hat J$ replaces each matrix slot by its two indices (same number of blocks); in particular $\|A_L\|_{\rm HS}\le2^{m/2}\lambda^c\|M_r\|_F$;
    \item $\max_Jq^{|J|/2}\|A_L\|_J\le2^{m/2}\lambda^cR^r(rq+\Sigma_r)^{r/2}$.
  \end{enumerate}
\end{lemma}
\begin{proof}
  (a) View $(G,g)$ as one standard Gaussian of $\mathrm{Sym}\oplus\mathcal H$. Its Wick tensor of order $j$, projected onto $m$ matrix and $c$ vector slots, is the product defining $U^L_{m,c}$. \cref{lem-dec}(a) and conditional Jensen give (a). (b) A unit test tensor on a block with $m_g$ matrix and $c_g$ vector slots maps to a coordinate tensor of norm at most $2^{m_g/2}\lambda^{c_g}$, since $\|X\|_F^2=2\|X\|^2_{\rm GOE}$. (c) Combine (b) with \cref{cor-reg}: $q^{l/2}\|A_L\|_J\le2^{m/2}\lambda^cR^r(rq)^{l/2}\Sigma_r^{(r-l)/2}$.
\end{proof}

\begin{proposition}[Mixed moments]\label{prop-msb}

  For all $m,c\ge0$ with $j=m+c\ge1$ and all $q\ge1$,

  \[
    \|U^L_{m,c}(\varphi)\|_q\le\lambda^c(Bj)^j\big(\sigma+12j^2+jq+j\big)^{m+c/2},\qquad B=C(B_1).
  \]

\end{proposition}
\begin{proof}
  If $c>0$ and $\lambda=0$, the chaos vanishes. Otherwise factor out $\lambda^c$ by scaling $L$ to norm one; when $c=0$, $L$ is irrelevant. Put $S=\sigma+12j^2+j$ and $F=\|M_r\|_F$, so that $\Sigma_r\le2S$.

  First let $q\ge2$. Apply \cref{lem-split} to the real and imaginary parts of the decoupled form in \cref{lem-mixed}:

  \[
    \|U\|_q\le C^jj^{j/2}\Big(q^{j/4}F+j^{(j-1)/2}(rq+\Sigma_r)^{r/2}\Big).
  \]

  \begin{itemize}
    \item \emph{Partition term.} It is at most $C^jj^j(S+jq)^{r/2}$, since $r\le2j$.
    \item \emph{Hilbert--Schmidt term.} \cref{lem-frob} and $\Sigma_r\ge r^2$ give $F\le C^jj^{j/2}S^{r/2-j/4}$. Since $r/2-j/4\ge0$, $j^{j/2}q^{j/4}F\le C^jj^jq^{j/4}S^{r/2-j/4}\le C^jj^j(S+jq)^{r/2}$.

  \end{itemize}
  For $1\le q<2$, use $\|U\|_q\le\|U\|_2$ and $S+2j\le2(S+jq)$, absorbing $2^{r/2}\le2^j$ into the constant.
\end{proof}

The separate Frobenius estimate is essential at the final step.
The Hilbert--Schmidt term in \cref{lem-split} carries $q^{j/4}$;
the $\Sigma_r^{r/4}$ dependence in \cref{lem-frob} lets it fit the
same bound as the other partition terms. This completes the bulk
estimate used in \cref{sec-cp-proof} and hence the proof of
\cref{thm-main}.

\section*{AI Disclosure}
The author used OpenAI Codex and Claude Code to assist with developing, checking, and
simplifying proof arguments; researching related work and verifying
citations; and drafting and revising the manuscript.
The author takes full responsibility for the correctness, originality,
and attribution of all content, including the proofs and references.

\bibliographystyle{alphaurl}
\bibliography{references}
\end{document}